\documentclass[12pt]{article}

\usepackage[utf8]{inputenc}
\usepackage[T1]{fontenc}
\usepackage{verbatim}
\usepackage[margin=3cm]{geometry}
\usepackage[algo2e,ruled,vlined]{algorithm2e}
\usepackage{algorithm}
\usepackage{here}
\usepackage{setspace}
\usepackage[normalem]{ulem}
\usepackage{algpseudocode}
\usepackage{dsfont}
\usepackage{xurl}
\usepackage[inline]{enumitem}
\usepackage{microtype}
\usepackage[affil-sl]{authblk}
\usepackage{graphicx}
\usepackage{nicematrix}
\usepackage{mathrsfs}
\usepackage{boldline,colortbl,float,import,pdflscape}

\usepackage{tikz}
\usetikzlibrary{tikzmark}
\usetikzlibrary{positioning}
\usetikzlibrary{graphs}
\usetikzlibrary{graphs.standard}
\usetikzlibrary{decorations.pathmorphing}
\usetikzlibrary{decorations.text}
\usetikzlibrary{patterns,shadings}
\usetikzlibrary{through,intersections,calc}
\usepackage{tikz-3dplot}

\usepackage{amssymb}
\usepackage{amsmath}
\usepackage{amsthm}
\usepackage{mathtools}
\usepackage{nicefrac}
\usepackage{amsfonts,latexsym,graphics}
\usepackage{color}
\usepackage[hidelinks]{hyperref}
\usepackage{multirow}
\usepackage[english]{babel}
\usepackage{epsfig}
\usepackage{bm}
\usepackage{makeidx}
\usepackage{subcaption}
\usepackage{adjustbox}
\usepackage{diagbox}
\usepackage{float}
\usepackage{stackengine} 

\usepackage{xcolor}
\usepackage{xspace}         

\newtheorem{theorem}{Theorem}[section]
\newtheorem{corollary}[theorem]{Corollary}
\newtheorem{lemma}[theorem]{Lemma}
\theoremstyle{definition}
\newtheorem{definition}[theorem]{Definition}
\newtheorem{proposition}[theorem]{Proposition}

\newtheorem{question}[theorem]{Question}

\newtheorem{observation}[theorem]{Observation}

\newtheorem{construction}[theorem]{Construction}

\newcommand*{\claimproofname}{Proof}

\usepackage{complexity}
\newclass{\Hard}{hard}
\newcommand{\NPh}{\NP-\Hard}
\newclass{\Complete}{complete}
\newclass{\Intermediate}{intermediate}
\newcommand{\NPc}{\NP-\Complete}
\newcommand{\NPi}{\NP-\Intermediate}
\newcommand{\GIc}{\GI-\Complete}
\newcommand{\pname}[1]{\textsc{#1}}

\newcommand{\probl}[3]{
  \begin{flushleft}
    \fbox{
      \begin{minipage}{0.97\linewidth}
        \noindent {\sc #1}\\
        {\bf Instance:} #2\\
        {\bf Question:} #3
      \end{minipage}}
    \medskip
  \end{flushleft}
}

\def\matrix0{{\mbox {\boldmath $O$}}}

\def\j{{\mbox{\boldmath $1$}}}
\def\vec0{\mbox{\bf 0}}

\newcommand\tran{\mkern-2mu\raise1.25ex\hbox{$\scriptscriptstyle\top$}\mkern-3.5mu}

	{\end{list}}

\DeclareRobustCommand{\VAN}[3]{#2} 

\SetArgSty{}
\SetKwFor{For}{for}{do}{}
\SetKwInOut{Input}{Input}
\SetKwInOut{Output}{Output}

\definecolor{border}{RGB}{100,120,170}
\definecolor{inner}{RGB}{240,245,255}
\colorlet{inner}{gray!20}
\colorlet{inner2}{gray!80}
\colorlet{border}{black!80}

\makeatletter
\let\c@figure\c@table
\let\ftype@figure\ftype@table
\let\ext@figure\ext@table

\newcommand{\gauss}[2]{\genfrac{[}{]}{0pt}{}{#1}{#2}\@ifnextchar\bgroup{_}{}}
\makeatother

\title{Computing fixed point free automorphisms \\of graphs}

\author{Aida Abiad\thanks{\texttt{a.abiad.monge@tue.nl}, 
Department of Mathematics and Computer Science, Eindhoven University of Technology, The Netherlands}
\thanks{Department of Mathematics and Data Science of Vrije Universiteit Brussels, Belgium}
\quad
Gabriel Coutinho\thanks{\texttt{gabriel@dcc.ufmg.br}, Department of Computer Science, Universidade Federal de Minas Gerais, Brazil} 
\quad
Emanuel Juliano\thanks{\texttt{emanuelsilva@dcc.ufmg.br}, Department of Computer Science, Universidade Federal de Minas Gerais, Brazil}
\hspace{3cm}
Vinicius F. dos Santos\thanks{\texttt{viniciussantos@dcc.ufmg.br}, Department of Computer Science, Universidade Federal de Minas Gerais, Brazil} 
\quad Sjanne Zeijlemaker\thanks{\texttt{s.zeijlemaker@tue.nl}, Department of Mathematics and Computer Science, Eindhoven University of Technology, The Netherlands}}
\date{}

\begin{document}

\maketitle

\begin{abstract}
In 1981, Lubiw proved that the \textit{fixed point free automorphism} problem (\pname{FPFAut}) is \NPc: given a graph $G$, determine whether there exists an automorphism $\phi : V(G) \to V(G)$ such that $\phi(u) \neq u$ for every $u \in V(G)$. We revisit this problem and prove that \pname{FPFAut} remains \NPc\ when restricted to split, bipartite, $k$-subdivided, and $H$-free graphs, if $H$ is not an induced subgraph of $P_4$. The class of $P_4$-free graphs receives the special name of cographs. We provide a polynomial time algorithm for three extensions of cographs: bounded modular-width graphs, tree-cographs and $P_4$-sparse graphs. Our approach uses the well known modular decomposition of graphs. As a consequence, we generalize a result of Abiad et. al. on the problem of computing $2$-homogeneous equitable partitions. 
\end{abstract}


 \section{Introduction}

For graphs $G$ and $H$, a bijection $\phi : V(G) \to V(H)$ is called an \emph{isomorphism} if $uv \in E(G) \iff \phi(u)\phi(v) \in E(H)$. The problem of deciding whether two graphs are isomorphic is called the graph isomorphism problem (\pname{GraphISO}). \pname{GraphISO} is a candidate to be an \NPi\, problem as no polynomial time algorithm to solve \pname{GraphISO} is known and it is not believed to be \NPh~\cite{schoning1988graph}. If a problem is polynomially equivalent to \pname{GraphISO} we say that the problem is \GIc~\cite{zemlyachenko1985graph}.

For a graph $G$, a bijection $\phi : V(G) \to V(G)$ is called an \emph{automorphism} if $uv \in E(G) \iff \phi(u)\phi(v) \in E(G)$. The problem of deciding whether a given graph possesses a non-trivial automorphism is \GIc\,~\cite{babai1996automorphism, mathon1979note}. An automorphism $\phi$ of a graph is called \emph{fixed point free} (FPF) if no vertex is mapped to itself. In 1981, Lubiw~\cite{Lubiw1981} proved that the fixed point free automorphism problem (\pname{FPFAut}) is \NPc.

\probl{FPFAut}{A graph $G$.}{Does $G$ admit an automorphism without fixed points?}

Lubiw's result shows that adding a simple restriction to the automorphism problem makes it NP-complete. This result has been used as a starting point for several polynomial-time reductions, highlighting the importance of FPF automorphisms in different areas. For example, motivated by eigenvalue bounds and applications in linear programming, Abiad et al.~\cite{abiad2022characterizing} introduce the concept of 2-homogeneous equitable partitions and use Lubiw’s result to prove that checking whether a graph has such partition is \NPc. In a more geometrical context, Manning \cite{manning1990geometric} shows that \pname{FPFAut} is related to checking whether a graph has an automorphism that can be drawn as a reflection. Schröder’s survey~\cite{schroder1999algorithms} presents several algorithms for checking whether an ordered set has the fixed point property, a problem also shown to be \NPc\ using Lubiw’s result.

We do not expect to find an efficient algorithm for \pname{FPFAut} in general. However, one could hope to find graph classes where the problem admits a polynomial solution. For example, Cameron~\cite{cameron2010complexity} observes that \pname{FPFAut} can be solved in polynomial time when restricted to vertex-transitive graphs, by pointing an old result of Jordan~\cite{jordan1872recherches, serre2003on} that applies the Burnside Lemma to show that there always exists an FPF element in a transitive group. The main goal of this paper is to identify broader families of graphs for which \pname{FPFAut} admits a polynomial-time solution.

We start by ruling out some well-known classes.

\begin{theorem}\label{thm:main_negative_result}
    FPFAut is NP-hard for split, bipartite, $k$-subdivided, and $H$-free graphs, if $H$ is not an induced subgraph of $P_4$.
\end{theorem}

The class of $P_4$-free graphs not covered by Theorem~\ref{thm:main_negative_result} receives the special name of \emph{cographs}. A polynomial time solution for \pname{FPFAut} restricted to cographs follows from a result of De Ridder and Bodlaender~\cite{de1999graph}. We extend this result to generalizations of cographs.

\begin{theorem}\label{thm:main_positive_result}
The problem of determining whether a graph $G$ contains a fixed point free automorphism can be solved in polynomial time when restricted to the class of bounded modular-width graphs, tree-cographs and $P_4$-sparse graphs.
\end{theorem}

The main idea for showing Theorem \ref{thm:main_positive_result} is to consider the modular decomposition of a graph~\cite{habib2010survey}. This follows a different strategy than the one employed by Klavík et al.~\cite{klavik2021graph} when studying automorphisms restricted to lists. Their results implies that \pname{FPFAut} can be solved in polynomial time for trees, planar graphs, interval graphs, circle graphs, permutation graphs, and bounded tree-width graphs. We complement the work of Klavík et al., by considering a simpler problem and solving it for more general graph classes.

One advantage of our approach is that it can be easily modified for finding FPF \emph{involutions}, that is, FPF automorphisms of order two. The problem of finding involutions without fixed points (\pname{FPFInv}) is also \NPh~\cite{Lubiw1981}. In fact, in the reduction used by Lubiw to prove that FPF automorphism is an \NPh\ problem, every FPF automorphism is an FPF involution. 

We extend Theorems~\ref{thm:main_negative_result} and \ref{thm:main_positive_result} to the \pname{FPFInv} problem and, as a byproduct, we show that determining whether a graph has an involution without fixed \emph{edges} is \NPc, providing an alternative proof to a result of Harary et al.~\cite{harary2002symmetric}. 

Our main motivation for considering involutions is to generalize a result of Abiad et al.~\cite{abiad2022characterizing}, who prove that a graph admits an FPF involution if and only if it admits an equitable partition with $\frac{n}{2}$ parts of size $2$, and they provide a polynomial-time algorithm to determine whether a cograph admits such a partition. By extending Theorem~\ref{thm:main_positive_result}, we also extend Abiad's result (see Theorem~\ref{thm:main_corollary_from_Aida}).

This paper is structured as follows. In Section \ref{sec:Negative_Results} we prove Theorem \ref{thm:main_negative_result} separately for split, bipartite, $k$-subdivided, and $H$-free graphs. Next, in Section \ref{sec:efficient} we provide an algorithm for \pname{FPFAut} in terms of the modular decomposition and prove its complexity. Then, in Section \ref{sec:positive_results}, we apply our algorithm and prove Theorem \ref{thm:main_positive_result} separately for graphs of bounded modular-width, tree-cographs and $P_4$-sparse graphs. In Section \ref{sec:involution}, we extend our results to the \pname{FPFInv} problem. Finally, in Section \ref{sec:conclusion} we conclude and discuss a few open questions.

\section{Hardness reductions} \label{sec:Negative_Results}

The isomorphism problem remain \GIc\, even under the restriction that the graphs $G$ and $H$ belong to some graph class $\mathcal{G}$. Booth and Colbourn (see survey~\cite{booth1979problems}) list several graph classes for which the restricted isomorphism problem is \GIc. In this section, we show that \pname{FPFAut} remains \NPh\ even when restricted to some well-known graph classes. 

It is worth mentioning that even though our results are inspired by the work on \pname{GraphISO}, the same constructions usually do not work for \pname{FPFAut}, as there is no clear one-to-one correspondence between the two problems. For example, \pname{GraphISO} is GI-complete for graphs with a unique center~\cite{booth1979problems}, while \pname{FPFAut} is trivial for this graph class.

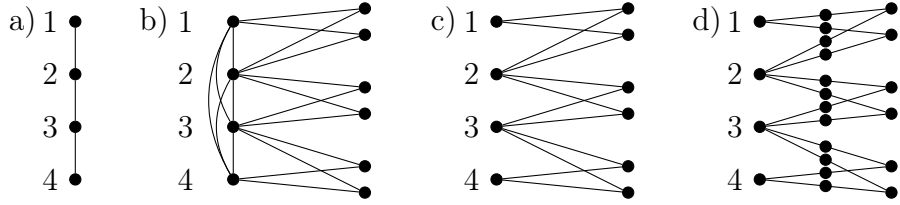
\begin{figure}[ht]
    \centering
    \begin{tikzpicture}[scale=0.7, every node/.style={circle, draw, fill=black, inner sep=1.5pt}]
    \node[draw=none,fill=none] at (-1, 3) {a)};
    \begin{scope}[xshift=0cm]
        \node[label=left:{1}] (1) at (0, 3) {};
        \node[label=left:{2}] (2) at (0, 2) {};
        \node[label=left:{3}] (3) at (0, 1) {};
        \node[label=left:{4}] (4) at (0, 0) {};
        
        \draw (1) -- (2);
        \draw (2) -- (3);
        \draw (3) -- (4);
    \end{scope}
    
    \node[draw=none,fill=none] at (1.5, 3) {b)};
    \begin{scope}[xshift=3cm]
        \node[label={[label distance=3mm]left:{1}}] (1) at (0, 3) {};
        \node[label={[label distance=3mm]left:{2}}] (2) at (0, 2) {};
        \node[label={[label distance=3mm]left:{3}}] (3) at (0, 1) {};
        \node[label={[label distance=3mm]left:{4}}] (4) at (0, 0) {};
        
        \node (121) at (2.5, 3.25) {};
        \node (122) at (2.5, 2.75) {};
        \node (231) at (2.5, 1.75) {};
        \node (232) at (2.5, 1.25) {};
        \node (341) at (2.5, 0.25) {};
        \node (342) at (2.5, -0.25) {};
        
        \draw (1) -- (2);
        \draw (1) edge[bend right] (3);
        \draw (1) edge[bend right] (4);
        \draw (2) -- (3);
        \draw (2) edge[bend right] (4);
        \draw (3) -- (4);
        
        \draw (1) -- (121);
        \draw (1) -- (122);
        \draw (2) -- (121);
        \draw (2) -- (122);
        \draw (2) -- (231);
        \draw (2) -- (232);
        \draw (3) -- (231);
        \draw (3) -- (232);
        \draw (3) -- (341);
        \draw (3) -- (342);
        \draw (4) -- (341);
        \draw (4) -- (342);
    \end{scope}

    \node[draw=none,fill=none] at (7, 3) {c)};
    \begin{scope}[xshift=8cm]
        \node[label={[label distance=0mm]left:{1}}] (1) at (0, 3) {};
        \node[label={[label distance=0mm]left:{2}}] (2) at (0, 2) {};
        \node[label={[label distance=0mm]left:{3}}] (3) at (0, 1) {};
        \node[label={[label distance=0mm]left:{4}}] (4) at (0, 0) {};
        
        \node (121) at (2.5, 3.25) {};
        \node (122) at (2.5, 2.75) {};
        \node (231) at (2.5, 1.75) {};
        \node (232) at (2.5, 1.25) {};
        \node (341) at (2.5, 0.25) {};
        \node (342) at (2.5, -0.25) {};
        
        \draw (1) -- (121);
        \draw (1) -- (122);
        \draw (2) -- (121);
        \draw (2) -- (122);
        \draw (2) -- (231);
        \draw (2) -- (232);
        \draw (3) -- (231);
        \draw (3) -- (232);
        \draw (3) -- (341);
        \draw (3) -- (342);
        \draw (4) -- (341);
        \draw (4) -- (342);
    \end{scope}

    \node[draw=none,fill=none] at (12, 3) {d)};
    \begin{scope}[xshift=13cm]
        \node[label={[label distance=0mm]left:{1}}] (1) at (0, 3) {};
        \node[label={[label distance=0mm]left:{2}}] (2) at (0, 2) {};
        \node[label={[label distance=0mm]left:{3}}] (3) at (0, 1) {};
        \node[label={[label distance=0mm]left:{4}}] (4) at (0, 0) {};
        
        \node (121) at (2.5, 3.25) {};
        \node (122) at (2.5, 2.75) {};
        \node (231) at (2.5, 1.75) {};
        \node (232) at (2.5, 1.25) {};
        \node (341) at (2.5, 0.25) {};
        \node (342) at (2.5, -0.25) {};

        \node (1211) at (1.25, 3.125) {};
        \node (1221) at (1.25, 2.875) {};
        \node (1212) at (1.25, 2.625) {};
        \node (1222) at (1.25, 2.375) {};
        \node (2312) at (1.25, 1.875) {};
        \node (2313) at (1.25, 1.375) {};
        \node (2322) at (1.25, 1.625) {};
        \node (2323) at (1.25, 1.125) {};
        \node (3413) at (1.25, 0.625) {};
        \node (3423) at (1.25, 0.375) {};
        \node (3414) at (1.25, 0.125) {};
        \node (3424) at (1.25, -0.125) {};
        
        \draw (1) -- (121);
        \draw (1) -- (122);
        \draw (2) -- (121);
        \draw (2) -- (122);
        \draw (2) -- (231);
        \draw (2) -- (232);
        \draw (3) -- (231);
        \draw (3) -- (232);
        \draw (3) -- (341);
        \draw (3) -- (342);
        \draw (4) -- (341);
        \draw (4) -- (342);
    \end{scope}
\end{tikzpicture}
    \caption{a) Base graph $G = P_4$. b) Split graph construction. c) Bipartite construction. d) Subdivision construction.}
    \label{fig:NP-H constructions}
\end{figure}

To prove that \pname{FPFAut} remains \NPh\ when restricted to a given graph class, we employ the following general approach.
Let $G = (V, E)$ be an arbitrary graph belonging to a class $\mathcal{F}_1$ for which \pname{FPFAut} is already known to be \NPh.
We construct a new graph $G' = (V', E')$ with $|G'| = |G|^{\mathcal{O}(1)}$ such that $G'$ belongs to another graph class $\mathcal{F}_2$, and $G$ admits an FPF automorphism if and only if $G'$ does.
This reduction implies that \pname{FPFAut} is also \NPh\ when restricted to graphs in~$\mathcal{F}_2$.
In this section, we describe the constructions for each graph class.

A \textit{split graph} is a graph whose vertex set can be partitioned into two classes, one that forms a clique and another that forms an independent set.

\begin{construction} \label{con:split}
    Let $G = (V, E)$ be a graph. We construct a split graph $G' = (V', E')$ as follows. Set $V' := V \cup \left(E \times \{1, 2\}\right)$ and connect vertices $u, v \in V'$ if either $u, v \in V$, or $u \in V$ and $v = \{u, w\}_i$ for some $w \in V$ and $i \in \{1, 2\}$ (in Figure~\ref{fig:NP-H constructions}\,(b) $G'$ is shown).
\end{construction}
\begin{theorem} \label{thm:split}
\pname{FPFAut} is \NPh\ for \emph{split graphs}.
\end{theorem}
\begin{proof}
    Let $G = (V, E)$ be a graph and let $G'$ be constructed following construction \ref{con:split}.
    
    Let $\phi$ be an automorphism of $G$ without fixed points. We construct a map $\phi'$ that is an automorphism of $G'$ as follows: $\phi'(v) := \phi(v)$ for all $v \in V$ and $\phi'(\{u, v\}_i) := \{\phi(u), \phi(v)\}_j$ with $i \neq j$. By construction, $\phi'$ fixes no vertex in $V$ and since $i \neq j$, no vertex in $E \times \{1, 2\}$ is fixed as well. As $\phi'$ preserves adjacency, it is an FPF automorphism. 
    
    Now, let $\phi'$ be an FPF automorphism of $G'$. The vertices in $V$ are the ones with the higher degree in $G'$, so $V$ is invariant under $\phi'$ and $\phi := \phi'|_{V}$ must be a FPF bijection. We check that it preserves adjacency in $G$. Let $\{u, v\} \in E$, there are exactly two vertices in $E \times \{1, 2\}$ adjacent to both $u$ and $v$, so we have $\phi'(\{u, v\}_i) = (\{\phi'(u), \phi'(v)\}_j)= \{\phi(u), \phi(v)\}_j$ for some $j \in \{1, 2\}$, and by construction, $\{\phi(u), \phi(v)\} \in E$.
\end{proof}

A \textit{chordal graph} is a graph that has no induced cycle of length greater than three. A \textit{perfect graph} is a graph with the property that, in every induced subgraph, the size of the largest clique equals the minimum number of colors in a proper coloring of the subgraph. There is a wide range of algorithmic results for both chordal and perfect graphs, and these classes are related by inclusion relations: split graphs $\subseteq$ chordal graphs $\subseteq$ perfect graphs.

\begin{corollary}
\pname{FPFAut} is \NPh\ for \emph{chordal} and \emph{perfect graphs}.
\end{corollary}

A \textit{bipartite graph} is a graph whose vertex set can be partitioned into two independent sets. It is known that \pname{GraphISO} remains \GIc\ when restricted to bipartite graphs, and the reduction follows by subdividing the edges of the original graph. In the context of FPF automorphisms, the same argument does not apply, since if an automorphism fixes an edge, the induced automorphism after subdivision will fix a vertex. We present an alternative construction.

\begin{construction} \label{con:bipartite}
    Let $G = (V, E)$ be a connected graph with $|V| \geq 3$. We construct a bipartite graph $G' = (V', E')$ as follows. Let $V' = V \cup \left(E \times \{1, 2\}\right)$, and connect vertices $u, v \in V'$ if $u \in V$ and $v = \{u, w\}_i$ for some $w \in V$ and $i \in \{1, 2\}$. Note that $V$ forms an independent set, so $G'$ is bipartite (in Figure~\ref{fig:NP-H constructions}\,(c), $G'$ is shown).
\end{construction}
\begin{theorem} \label{thm:bipartite}
\pname{FPFAut} is \NPh\ for \emph{bipartite graphs}.
\end{theorem}
\begin{proof}
Let $G = (V, E)$ be a connected graph with $|V|\geq 3$ and let $G'$ be constructed following Construction~\ref{con:bipartite}.

As in the proof of Theorem~\ref{thm:split}, it is easy to see that any FPF automorphism of $G$ induces an FPF automorphism of $G'$. For the converse, observe that because $G$ is connected, $G'$ is also connected and there exists a unique bipartition of the vertices of $G'$. This bipartition fixes the set $V$ as a part. Because $|V|\geq 3$, there exists a vertex $u \in V$ with degree different from $2$ in $G'$ and the set of vertices whose distance from $u$ is even forms exactly the set $V$. Therefore, any FPF automorphism must fix $V$ as a set, inducing an FPF automorphism in $G$.
\end{proof}

Let $G$ be a graph and $S(G)$ denote the subdivision of $G$, that is, we replace each edge of $G$ by two edges with a vertex in the middle. The class of \emph{subdivided graphs} is the class of all graphs obtained by subdividing some other graph. FPF automorphisms are particularly interesting for this graph class because if one edge is fixed by an automorphism, then the extension of this automorphism is not fixed point free in $S(G)$. 

Fortunately, in the bipartite construction of Theorem~\ref{thm:bipartite}, a graph has an FPF automorphism if and only if it has a fixed \emph{edge} free automorphism. We denote by $S_k(G)$ the $k$-subdivision of $G$, that is, the graph obtained after replacing each edge of $G$ by a path on $k+2$ vertices. Notice that with this definition, $S(G) = S_1(G)$. The class of \emph{k-subdivided graphs} is the class of all graphs obtained from the $k$-subdivision of some other graph.

\begin{construction} \label{con:subdvided}
Let $G = (V, E)$ be a connected graph with $|V| \ge 3$, and let $G'$ be the bipartite graph constructed according to Construction~\ref{con:bipartite}. We then construct the $k$-subdivided graph $S_k(G')$
(In Figure~\ref{fig:NP-H constructions}\,(d), the graph $S(G')$ is shown).
\end{construction}
\begin{theorem}\label{thm:subdvided}
    \pname{FPFAut} is \NPh\ for \emph{$k$-subdivided graphs}.
\end{theorem}
\begin{proof}
Let $G = (V, E)$ be a connected graph with $|V| \geq 3$. and let $G'$ be constructed following Construction~\ref{con:bipartite}. Note that every FPF automorphism of $G'$ is also a fixed edge free automorphism because $G'$ is bipartite and there is no automorphism swapping both parts. We claim that the graph $S_k(G')$ from Construction~\ref{con:subdvided} has an FPF automorphism if and only if $G'$ has an FPF automorphism.

Let $\phi'$ be an automorphism of $G'$ without fixed points. Let $(u, v)_{i, d}$ denote the vertex at distance $d$ from $u$ in the new subdivided path from $u$ to $\{u, v\}_i$. We construct a map $\phi_k$ that is an automorphism of $S_k(G')$. The restriction of $\phi_k$ on $V(G')$ is equal to $\phi'$, and for the new vertices, we set $\phi_k((u, v)_{i, d}) := (u', v')_{j, d}$ given that $\phi'(\{u, v\}_i) = \{u', v'\}_j$. This map preserves the distance between the vertices and, since $\phi'$ fixes no edge, $\phi_k$ is an FPF automorphism.

Now, let $\phi_k$ be an automorphism with no fixed point of $S_k(G')$. Let $u$ be a vertex of degree different from $2$. We know that such a vertex exists and belongs to the set $V$ because $G$ is connected and $|V| \geq 3$. The set of vertices whose distance from $u$ is a multiple of $k+1$ forms exactly the set $V(G')$, so the vertices of $G'$ are fixed as a set by $\phi_k$. Since two vertices of $G'$ are adjacent if and only if they are at distance $k+1$ in $S_k(G')$, we conclude that $\phi' := \phi_{k|V(G')}$ is an FPF automorphism of $G'$.
\end{proof}

In Section~\ref{sec:positive_results}, we discuss a polynomial-time algorithm for \pname{FPFAut} on cographs, the class of graphs with no induced $P_4$. We show that by forbidding a different graph, the problem of finding an FPF automorphism is \NPc. Let $H$ be a graph and denote by $H$-free graphs the family of graphs having no induced subgraph isomorphic to $H$.

\begin{theorem} \label{thm:h-free}
    \pname{FPFAut} is \NPh\ for $H$-free graphs, if $H$ is not an induced subgraph of $P_4$.
\end{theorem}
\begin{proof}
We adapt an argument from~\cite{booth1979problems}. Let $H$ be a graph with $k$ vertices that is not an induced subgraph of $P_4$. Notice that either $H$ has a cycle or $\overline{H}$ has a cycle, as every automorphism of the graph is also an automorphism of its complement, \pname{FPFAut} is \NPh\ for $H$-free graphs if and only if \pname{FPFAut} is \NPh\ for $\overline{H}$-free graphs. Assume without loss of generality that $H$ has a cycle. Since the girth of any $k$-subdivided graph is greater than $k$, the class of $H$-free graphs contains the class of $k$-subdivided graphs and the result follows by Theorem~\ref{thm:subdvided}.
\end{proof}

 \section{Algorithm for \pname{FPFAut}} \label{sec:efficient}

\subsection{Modular Decomposition}

We follow the notation from Habib and Paul~\cite{habib2010survey} for modular decompositions. Let $M \subseteq V$ be a set of vertices of a graph $G = (V, E)$ and $u$ be a vertex of $V \setminus M$. Vertex $u$ \emph{splits} $M$ if there exists $v, w \in M$ such that $\{u, v\} \in E$ and $\{u, w\} \not \in E$. If $u$ is not a splitter of $M$, then $M$ is \emph{uniform} with respect to $u$. We say that $M$ is a \emph{module} if it is uniform with respect to any $u \not \in M$. The union of connected components (or of co-connected components) are examples of modules. 

Two modules $M$ and $M'$ \emph{overlap} if $M \cap M' \neq \emptyset$, $M\setminus M' \neq \emptyset$ and $M'\setminus M \neq \emptyset$. A module is \emph{strong} if it does not overlap any other module. A connected component is an example of a strong module.

Let $\mathcal{P} = \{M_1, \dots, M_k\}$ be a partition of the vertex set of a graph $G$. If, for all $1 \leq i \leq k$, $M_i$ is a module of $G$, then $\mathcal{P}$ is a \emph{modular partition} of $G$.

A module $M$ is \emph{maximal} with respect to a set $S$ of vertices, if $M \subset S$ and there is no module $M'$ such that $M \subset M' \subset S$. A non-trivial modular partition $\mathcal{P} = \{M_1, \dots, M_k\}$ which only contains maximal strong modules is a maximal modular partition. If $G$ (resp. $\overline{G}$) is not connected, then the elements of its connected (resp. co-connected) components are the elements of the maximal modular partition. The maximal modular partition can be computed in linear time.

\begin{lemma}[\cite{mcconnel1994linear}, Lemma 5.2]\label{lem:modules_complexity}
    A maximal modular partition of a graph $G$ can be computed in time $\mathcal{O}(|V(G)|+|E(G)|)$.
\end{lemma}

Let $M$ and $M'$ be disjoint modules. We say that $M$ and $M'$ are \emph{adjacent} if every vertex of $M$ is adjacent to every vertex of $M'$ and \emph{non-adjacent} if every vertex of $M$ is non-adjacent to every vertex of $M'$. Note that two disjoint modules are either adjacent or non-adjacent. To a modular partition $\mathcal{P} = \{M_1, \dots, M_k\}$ of a graph $G$, we associate a \emph{quotient graph} $G_{/\mathcal{P}}$, whose vertices are in one-to-one correspondence with the parts of $\mathcal{P}$. Two vertices $v_i$ and $v_j$ of $G_{/\mathcal{P}}$ are adjacent if and only if the corresponding modules $M_i$ and $M_j$ are adjacent in $G$. In case the partition $\mathcal{P}$ is not specified, we assume $\mathcal{P}$ is the maximal modular partition. 

\begin{definition}
    Let $G$ be a graph, let $\mathcal{P} = \{M_1, \dots, M_k\}$ be the maximal modular partition of $G$ and let $G_{/\mathcal{P}}$ be the associated quotient graph. Let $c : V(G_{/\mathcal{P}}) \to [k]$ be a coloring of the vertices of $G_{/\mathcal{P}}$ such that two vertices $v_i, v_j \in V(G_{/\mathcal{P}})$ receive the same color $c(v_i) = c(v_j)$ if and only if $G[M_i] \cong G[M_j]$. We refer to the pair $(G_{/\mathcal{P}}, c)$ as the \emph{quotient colored graph} associated with $\mathcal{P}$. 
\end{definition}

\subsection{PFPFAut}

Our main strategy is to explore the fact that the quotient graph $G_{/\mathcal{P}}$ can be simpler than the original graph, and due to this simpler structure, we can solve a harder variation of \pname{FPFAut}.

The set $\{1, 2, \dots, n\}$ is abbreviated as $[n]$. The set of automorphisms of $G$ is denoted by $\text{Aut}(G)$. Let $c : V(G) \to [k]$ be a coloring of $V(G)$ with $k$ colors, the set of color preserving automorphisms is denoted as $\text{Aut}(G, c)$.

\begin{definition}
    Let $G$ be a graph, let $c: V(G) \to [k]$ be a coloring of the vertices of $G$ and let $b : V(G) \to \{0, 1\}$ be a boolean function on the vertices of $G$. Let $\phi \in \text{Aut}(G, c)$, we say that $\phi$ is \emph{partial fixed point free with respect to $b$} (PFPF$_{b}$) if $\phi(v) = v$ only if $b(v) = 1$. 
\end{definition}

The boolean function $b$ above serves to allow some fixed points. The decision problem associated is the \pname{PFPFAut} problem.

\probl{PFPFAut}{A graph $G$, a coloring $c$, and a boolean function $b$.}{Does $(G, c)$ admit an PFPF$_b$ automorphism?}

We point that \pname{PFPFAut} can be reduced to automorphism restricted to lists~\cite{klavik2021graph}, a particular type of automorphism where each vertex $v \in V(G)$ is associated with a list $\mathcal{L}(v) \subseteq V(G)$, and the automorphism $\phi$ satisfies the property that $\phi(v) \in \mathcal{L}(v)$. The reduction follows by setting $\mathcal{L}(v) := \{w : c(w) = c(v)\}$ in case $b(v) = 1$ and $\mathcal{L}(v) := \{w : c(w) = c(v), w \neq v\}$, in case $b(v) = 0$. But for our applications, the \pname{PFPFAut} problem will be solved by simpler algorithms than the ones present in~\cite{klavik2021graph}. This approach is going to be useful in Section \ref{sec:involution} when we modify our algorithm to deal with FPF involutions.

\subsection{Algorithm for \pname{FPFAut}}

\begin{algorithm2e}[ht]
      \SetAlgoLined
      \SetKwFunction{FhasFPFAut}{hasFPFAutomorphism}
      \SetKwProg{Fn}{}{}{}
      \Input{A graph $G$.}
      \Output{Does $G$ admit an FPF automorphism?}

      \Fn{\FhasFPFAut{$G$}}{
        \uIf {$|V(G)| = 1$} {
            \Return false\\
        }
        let $\mathcal{P} = \{M_1, \dots M_k\}$ be the maximal modular partition of $G$.\\
        let $(G_{/\mathcal{P}}, c)$ be a quotient colored graph associated with $\mathcal{P}$.\\
        let $b \in \{0, 1\}^k$ with $b_i = 1$ if and only if \FhasFPFAut{$G[M_i]$}=true.\\
        \uIf{($G_{/\mathcal{P}}$, $c$) has an PFPF$_b$ automorphism}{
            \Return true\\
        }
        \Return false\\
      }
      \caption{\FuncSty{hasFPFAutomorphism}}
      \label{alg:fpf_aut}
\end{algorithm2e}

\begin{theorem}\label{thm:fpf_aut_alg}
    Let $G$ be a graph. Algorithm~\ref{alg:fpf_aut} returns ``true'' if and only if $G$ contains a fixed point free automorphism.
\end{theorem}

We prove the two directions of theorem~\ref{thm:fpf_aut_alg} separately.

\begin{lemma} \label{lem:fpf_aut_true_implies_correct}
    If \FhasFPFAut{$G$} returns ``true'', then $G$ has an FPF automorphism.
\end{lemma}
\begin{proof}
    We prove the claim by constructing an FPF automorphism $\phi$ from Algorithm~\ref{alg:fpf_aut}. The proof follows by induction on $|V(G)|$. If $|V(G)| = 1$, the algorithm returns ``false''. Assume $|V(G)| \geq 2$ and that \FhasFPFAut{$G$} returns ``true''. 

    Let $\mathcal{P} = \{M_1, \dots M_k\}, (G_{/\mathcal{P}}, c)$ and $b$ be defined as in Algorithm~\ref{alg:fpf_aut}. If $b_i = 1$, we know, by induction, that $G[M_i]$ has an FPF automorphism, say $\phi_i$. We also know that there exists $\psi \in \text{Aut}(G_{/\mathcal{P}}, c)$ such that every vertex $i$ of $G_{/\mathcal{P}}$ fixed by $\psi$ satisfies $b_i = 1$. Consider the extension of $\psi$ to an automorphism $\phi'$ of $G$ such that the map $v_i$ to $\psi(v_i)$ in $G_{/\mathcal{P}}$ is extended to a map from $G[M_i]$ to $G[M_{\psi(i)}]$, which by the coloring $c$ we know exists. Let $\phi$ be the composition of $\phi'$ with the FPF automorphisms $\phi_i$. By construction, $\phi$ is an automorphism. Let $v \in M_i$ for some $i$, if $\phi(v) \in M_i$, this means that $\phi(v) = \phi_i(v) \neq v$. Thus, $\phi$ is an FPF automorphism.
\end{proof}

Before proving the other implication, we need the following lemma.

\begin{lemma}\label{lem:modular_aut}
    Let $G$ be a graph, let $u, v \in V(G)$, and let $M_u$ and $M_v$ be the maximal strong modules that contain $u$ and $v$, respectively. If $\phi$ is an automorphism of $G$ such that $\phi(u) = v$, then $\phi(M_u) = M_v$.
\end{lemma}
\begin{proof}
    First, notice that $\phi(M_u)$ must be a module, as the neighborhood between $z$ and $M_u$ is preserved by $\phi$ for every $z \not \in M_u$, so $\phi(z)$ is uniform w.r.t $\phi(M_u)$. Moreover, $\phi(M_u)$ is strong, as $M$ overlapping $\phi(M_u)$ implies that $\phi^{-1}(M)$ overlaps $M_u$. Finally, $\phi(M_u)$ is maximal, since if there exists a vertex $z$ such that $\phi(M_u) \cup \{z\}$ is a module, then $M_u \cup \{\phi^{-1}(z)\}$ would be a module. We conclude that $\phi(M_u)$ is the maximal strong module that contains $v$, that is, $M_v$.
\end{proof}

\begin{lemma}
    If $G$ has an FPF automorphism then \FhasFPFAut{$G$} returns ``true''.
\end{lemma}
\begin{proof}
    We prove the claim by induction. If $|V(G)| = 1$ then $G$ has no FPF automorphism. Assume $|V(G)| \geq 2$ and that $G$ has an FPF automorphism $\phi$. By Lemma~\ref{lem:modular_aut}, $\phi$ induces an automorphism $\psi$ in $(G_{/\mathcal{P}}, c)$. Thus, every module $M_i$ fixed by $\psi$ must have an FPF automorphism $\phi_i$. By induction, $b_i$ has value $1$ for these modules and \FhasFPFAut{$G$} returns ``true''. 
\end{proof}

\subsection{Complexity}

We upper bound the running time of Algorithm~\ref{alg:fpf_aut} for different graph classes that can be characterized in terms of their modular decompositions. 

The recursion tree of Algorithm~\ref{alg:fpf_aut} corresponds to the standard modular tree, and its size can be bounded as follows.

\begin{lemma} \label{lem:tree_size}
    Let $G$ be a graph. The recursion tree of Algorithm~\ref{alg:fpf_aut} has at most $\mathcal{O}(|V(G)|)$ nodes.
\end{lemma}

A graph class $\mathcal{G}$ is called \emph{hereditary} if $G[S] \in \mathcal{G}$ for any $G \in \mathcal{G}$ and $S \subseteq V(G)$. To construct the coloring graph $(G_{/\mathcal{P}}, c)$, it is necessary to decide for $|\mathcal{P}|^2$ pairs of induced subgraphs of $G$ if they are isomorphic. So, if $G$ belongs to a hereditary graph class $\mathcal{G}$, this complexity depends on the complexity of the isomorphism problem between graphs in $\mathcal{G}$.

\begin{definition}
    Let $\mathcal{G}$ be a graph class. We denote by $f_{\text{iso}}^{\mathcal{G}}$ the time complexity of deciding if two graphs in $\mathcal{G}$ are isomorphic.
\end{definition}

\begin{definition}
    Let $\mathcal{G}$ be a graph class. We denote by $\mathcal{G}_{/\mathcal{P}} := \{G_{/\mathcal{P}} : G \in \mathcal{G}\}$.
\end{definition}

Note that the complexity of deciding whether $(G_{/\mathcal{P}}, c)$ contains an $\text{PFPF}_b$ automorphism is bounded in terms of the graph class $\mathcal{G}_{/\mathcal{P}}$.

\begin{definition}
    Let $\mathcal{G}$ be a graph class. Let $f_{\text{pfpf}}^{\mathcal{G}}$ denote the time complexity of deciding if a graph $G \in \mathcal{G}$ admits an $\text{PFPF}_b$ automorphism, for an arbitrary coloring $c : V(G) \to [|V(G)|]$ and arbitrary boolean function $b : V(G) \to \{0, 1\}$.
\end{definition}

\begin{lemma} \label{lem:fpf_aut_complexity}
 Let $\mathcal{G}$ be an hereditary graph class. The complexity of deciding whether an $n$-vertex graph $G \in \mathcal{G}$ contains an FPF automorphism is bounded above by 
 \[n^{\mathcal{O}(1)} \left(f_{\text{iso}}^{\mathcal{G}}(n) +f_{\text{pfpf}}^{\mathcal{G}_{/\mathcal{P}}}(\text{n})\right).
 \]
\end{lemma}

\section{Positive results} \label{sec:positive_results}

In this section we show that \pname{FPFAut} can be solved in polynomial time for bounded modular-width graphs, tree-cographs and $P_4$-sparse graphs.

\subsection{Modular-width} \label{subsec:modular_width}
We consider graphs that can be obtained from an algebraic expression that uses the following operations:

\begin{enumerate}[label=(\roman*)]
\item Create an isolated vertex;
\item The disjoint union of $2$ graphs, i.e., the \emph{disjoint union} of $2$ graph $G_1$ and $G_2$ is the graph with vertex set $V(G_1) \cup V(G_2)$ and edge set $E(G_1) \cup E(G_2)$; 
\item The complete join of $2$ graphs, i.e., the \emph{complete join} of $2$ graphs $G_1$ and $G_2$ is the graph with vertex set $V(G_1) \cup V(G_2)$ and edge set $E(G_1) \cup E(G_2) \cup \{ \{v,w\} : v \in V(G_1) \textup{ and }w \in V(G_2) \}$; 
\item The substitution operation with respect to some graph $G$ with vertices $v_1,\dots,v_n$, i.e., for graphs $G_1,\dots,G_n$ the \emph{substitution} of the vertices of $G$ by the graphs $G_1,\dots,G_n$, denoted by $G(G_1,\dots,G_n)$, is the graph with vertex set $\bigcup_{1\leq i \leq n}V(G_i)$ and edge set $\bigcup_{1 \leq i \leq n}E(G_i) \cup \{ \{u,v\} : u \in V(G_i) \textup{ and } v \in V(G_j) \textup{ and } \{v_i, v_j\} \in E(G)\}$. 
\end{enumerate}
Let $A$ be an algebraic expression that uses only the operations
(i)--(iv). We define the \emph{width} of $A$ as the maximum number of
operands used by any occurrence of the operation (iv) in $A$.
The \emph{modular-width} of a graph $G$, denoted
$\text{mw}(G)$, is the least integer $k$ such that $G$ can be obtained from such an algebraic expression of width at most $k$~\cite{gajarsky2013parameterized}. A graph $G$ has bounded modular-width if there exists a fixed constant $k$ such that $\mathrm{mw}(G) \le k$. Cographs are the class of graphs with modular-width $0$~\cite{CLHB1981}.

Let $\mathcal{K}, \overline{\mathcal{K}}$ and $\mathcal{G}_{\leq k}$ denote the class of complete graphs, empty graphs, and graphs with at most $k$ vertices, respectively.
\begin{observation} \label{obs:quotient_mw}
    Let $\mathcal{G}$ be the class of graphs with modular-width at most $k$ and let $\mathcal{G}_{/\mathcal{P}}$ be the associated quotient graph class. Then, $\mathcal{G}_{/\mathcal{P}} \subseteq \mathcal{K} \cup \overline{\mathcal{K}} \cup \mathcal{G}_{\leq k}$. 
\end{observation}

In order to apply Lemma~\ref{lem:fpf_aut_complexity}, we need to bound the functions $f_{\text{pfpf}}^{\mathcal{G}_{/\mathcal{P}}}$ and $f_{\text{iso}}^{\mathcal{G}}$. 

\begin{lemma} \label{lem:pfpf_complete}
    If $G \in \mathcal{K} \cup \overline{\mathcal{K}}$, then \pname{PFPFAut} can be solved in polynomial time.
\end{lemma}
\begin{proof}
    Let $G \cong K_n$, let $c : V(G) \to [k]$ be a coloring of the vertices of $G$ and let $b : V(G) \to \{0, 1\}^k$ be a boolean function of the vertices of $G$. Let $S_i = \{v \in V(G) : c(v) = i\}$ be the color class of color $i$. Note that if $G$ has an PFPF$_b$ automorphism, then $b_i = 1$ every time $S_i$ has size $1$. We claim that it is sufficient to check this condition. Let $\phi_i$ be a permutation that fixes every vertex outside of $S_i$ and if $|S_i| \geq 2$, $\phi_i$ is FPF in $S_i$. Let $\phi = \phi_1 \circ \dots \circ \phi_k$. If $b_i = 1$ every time $S_i$ has size $1$, then $\phi$ is an PFPF$_b$ automorphism of $G$.

    The claim holds in case $G \cong \overline{K_n}$, because every automorphism of a graph is also an automorphism of it's complement.
\end{proof}

\begin{lemma}\label{lem:iso_mw}
    Let $G$ and $H$ be graphs with bounded modular-width. The problem of determining whether $G$ is isomorphic to $H$ can be solved in polynomial time.
\end{lemma}
\begin{proof}
    Let $\text{rw}(G)$ be the rank-width of $G$. The rank-width of a graph lower bounds its modular-width, that is, there exists a computable function $g$ such that $\text{rw}(G) \leq g(\text{mw}(G))$~\cite{gajarsky2013parameterized}. Since \pname{GraphISO} can be solved in polynomial time for graphs with bounded rank-width~\cite{grohe2015isomorphism}, the result follows.
\end{proof}

Now we have all that is necessary to apply Lemma~\ref{lem:fpf_aut_complexity}.

\begin{theorem}\label{thm:fpfaut_mw}
    Let $G$ be a graph with bounded modular-width. The problem of determining whether $G$ contains a fixed point free automorphism can be solved in polynomial time.
\end{theorem}
\begin{proof}
    From Observation~\ref{obs:quotient_mw} and Lemma~\ref{lem:pfpf_complete} the function $f_{\text{pfpf}}^{\mathcal{G}_{/\mathcal{P}}}(|V(G)|)$ is bounded above by $|V(G)|^{\mathcal{O}(1)}$. From Lemma~\ref{lem:iso_mw} the function $f_{\text{iso}}^{\mathcal{G}}(|V(G)|)$ is bounded by $|V(G)|^{\mathcal{O}(1)}$. Therefore, by Lemma~\ref{lem:fpf_aut_complexity}, Algorithm~\ref{alg:fpf_aut} runs in time $|V(G)|^{\mathcal{O}(1)}$.
\end{proof}

A natural question would be whether the degree of the polynomial of the algorithm implied by Theorem~\ref{thm:fpfaut_mw} does not depend on the modular-width. We remark that the only step where this is not the case is the isomorphism test.

 \subsection{Tree-cographs}

Cographs can be defined recursively by starting with a set of isolated vertices and repeatedly applying disjoint union and complementation operations. If we replace the isolated vertices by trees, we obtain the following extension of cographs.

\begin{definition}
    The class of \emph{tree-cographs} is defined by the following recursive rules~\cite{tinhofer1988strong}. 
    \begin{enumerate}[label=(\roman*)]
        \item Every tree is a tree-cograph.
        \item If $G_1, \dots, G_k$ are tree-cographs, then so is $G_1 \cup \dots \cup G_k$.
        \item If $G$ is a tree-cograph, then so is its complement~$\overline{G}$.
    \end{enumerate}
\end{definition}

Tree-cographs generalize both the class of cographs and trees.

Let $\mathcal{T}$ and $\overline{\mathcal{T}}$ denote the class of trees and complement of trees, respectively. 
\begin{observation} \label{obs:quotient_tree_cographs}
    Let $\mathcal{G}$ be the class of tree-cographs and let $\mathcal{G}_{/\mathcal{P}}$ be the associated quotient graph class. Then, $\mathcal{G}_{/\mathcal{P}} \subseteq \mathcal{K} \cup \overline{\mathcal{K}} \cup \mathcal{T} \cup \overline{\mathcal{T}}$. 
\end{observation}

\begin{lemma} \label{lem:pfpf_trees}
    If $G \in \mathcal{T} \cup \overline{\mathcal{T}}$, then \pname{PFPFAut} can be solved in polynomial time.
\end{lemma}
\begin{proof}
    We first solve the problem for rooted trees. Let $T_r$ be a tree with root $r$, let $c : V(T_r) \to [n]$ a coloring and $b : V(T_r) \to \{0, 1\}$ a boolean function. Since the tree is rooted, we assume that $b(r) = 1$. We construct an empty graph $G'$ with vertex set $\{c_1, \dots, c_k\}$ corresponding to the children of $r$. We define a new boolean function $b' : V(G') \to \{0, 1\}$ such that $b(c_i) = 1$ if and only if $(T_{c_i}, c)$ admits an $\text{PFPF}_b$ automorphism. Finally, we construct a new coloring $c': V(G') \to [k]$ such that $c'(c_i) = c'(c_j)$ if and only if $(T_{c_i}, c) \cong (T_{c_j}, c)$. From construction it follows that $(T_r, c)$ has an $\text{PFPF}_b$ automorphism if and only if $(G', c')$ has an $\text{PFPF}_{b'}$ automorphism. Therefore, we have reduced the problem to solving \pname{PFPFAut} for the class of empty graphs, which we know from~\ref{lem:pfpf_complete} that can be solved in polynomial time.
    
    Now we return to non-rooted trees. Let $T$ be a tree, $c$ a coloring and $b$ a boolean function. If $T$ has a unique center $r_1$ then $(T, c)$ has an $\text{PFPF}_b$ automorphism if and only if $(T_{r_1}, c)$ has an $\text{PFPF}_b$ automorphism. If $T$ has two centers $r_1$ and $r_2$ but $b(r_1) \wedge b(r_2) = 1$, then $(T, c)$ has an $\text{PFPF}_b$ automorphism if and only if both $(T_{r_1}, c)$ and  $(T_{r_1}, c)$ have $\text{PFPF}_b$ automorphisms. Otherwise, $T$ has two centers and $b(r_1) \wedge b(r_2) = 0$, but in this case $(T, c)$ has an $\text{PFPF}_b$ automorphism if and only if there exists a color preserving isomorphism between $(T_{r_1}, c)$ and $(T_{r_2}, c)$, which can be computed in polynomial time. 
\end{proof}

Tree-cographs are uniquely determined by a labeled tree similar to the recursion tree of Algorithm~\ref{alg:fpf_aut}. The isomorphism problem can be solved in polynomial time for this graph class. 

\begin{lemma}[\cite{tinhofer1988strong}]\label{lem:iso_tree_cographs}
    Let $G$ and $H$ be tree-cographs. The problem of determining whether $G$ is isomorphic to $H$ can be solved in polynomial time.
\end{lemma}

\begin{theorem}\label{thm:fpfaut_tree_cograph}
    Let $G$ be a tree-cograph. The problem of determining $G$ contains a fixed point free automorphism can be solved in polynomial time.
\end{theorem}

\subsection{$P_4$-sparse graphs} \label{subsec:p4sparse}

The class of $P_4$-sparse graphs is a generalization of cographs, in which every set of five vertices induces at most one $P_4$. $P_4$-sparse graphs can be represented in terms of the modular decomposition, for which we will need the following definition from~\cite{jamison1992tree}.

A \emph{spider} is a graph $G = (V, E)$ whose vertex set can be partitioned as $V = S \cup K \cup R$ such that
\begin{itemize}
    \item $|S| = |K| \geq 2$, with $G[S]$ being an empty graph and $G[K]$ being a complete graph;
    \item every $w \in R$ is adjacent to all vertices in $K$ and to none in $S$;
    \item there exists a bijection $f: S \to K$ such that for all $s \in S$, $G(s) \cap K = \{f(s)\}$ or $G(s) \cap K = K \setminus \{f(s)\}$.
\end{itemize}
This vertex decomposition is unique for a given spider. 
\begin{figure}[h]
    \centering
    \begin{tikzpicture}[every node/.style={circle, draw, fill=black, inner sep=1.5pt}, scale=0.7]

\def\rot{22.5}

\foreach \i in {1,...,8} {
  \node (b\i) at ({cos(\rot + 360/8*(\i-1))}, {sin(\rot + 360/8*(\i-1))}) {};
}

\foreach \i in {1,...,8} {
  \path (b\i) -- ++({cos(\rot + 360/8*(\i-1))}, {sin(\rot + 360/8*(\i-1))}) coordinate (tmp);
  \node (l\i) at ($(b\i)!0.8!(tmp)$) {};
  \draw (b\i) -- (l\i);
}

\foreach \i in {1,...,8} {
  \foreach \j in {1,...,8} {
    \ifnum\i<\j
      \draw (b\i) -- (b\j);
    \fi
  }
}

\node (h) at (-1.5, 0) {};

\foreach \i in {1,...,8} {
  \draw (h) -- (b\i);
}

\end{tikzpicture}
    \caption{Cyclops spider with $|K|=|S|=8$.}
    \label{fig:spider}
\end{figure}
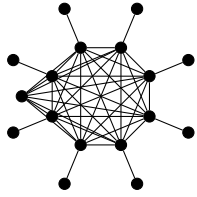

\begin{proposition}[\cite{jamison1992tree}]
    A $P_4$-sparse graph can be characterized recursively using the following four rules:
    \begin{enumerate}[label=(\roman*)]
        \item A graph on a single vertex is $P_4$-sparse.
        \item If $G_1, \dots, G_k$ are $P_4$-sparse, then so is $G_1 \cup \dots \cup G_k$.
        \item If $G$ is $P_4$-sparse, then so is its complement $\overline{G}$.
        \item If $H$ is $P_4$-sparse, and $G$ is a spider with $G[R] = H$, then $G$ is $P_4$-sparse.
    \end{enumerate}
\end{proposition}

In the case where $|R| = 1$, we say that $G$ is a \emph{cyclops spider}. We denote by $\mathcal{S_C}$ be the class of cyclops spiders.
\begin{observation} \label{obs:quotient_p4_sparse}
    Let $\mathcal{G}$ be the class of $P_4$-sparse graphs and let $\mathcal{G}_{/\mathcal{P}}$ be the associated quotient graph class. Then, $\mathcal{G}_{/\mathcal{P}} \subseteq \mathcal{K} \cup \overline{\mathcal{K}} \cup \mathcal{S_C}$.
\end{observation}

\begin{lemma}  \label{lem:pfpf_spider}
If $G \in\mathcal{S_C}$, then \pname{PFPFAut} can be solved in polynomial time.
\end{lemma}
\begin{proof}
    Let $G = (S\cup K \cup R, E)$ be a cyclops spider with bijection $f : S \to K$. Let $c : V(G) \to [n]$ a coloring and $b : V(G) \to \{0, 1\}$ a boolean function. We reduce the problem to solving \pname{PFPFAut} for the class of complete graphs, which we know from~\ref{lem:pfpf_complete} that can be solved in polynomial time. Let $R = \{h\}$ be the head of the spider, first, note that it is necessary that $b(h) = 1$ for an $\text{PFPF}_b$ automorphism to exits, so we assume that this is the case. We construct a new graph $G' \cong G[K]$, a new coloring $c' : K \to [n] \times [n]$ such that $c'(v) := (c(v), c(f^{-1}(v)))$ and a new boolean function $b' : K \to \{0, 1\}$ such that $b'(v) := b(v) \wedge b(f^{-1}(v))$. Note that every automorphism $\phi \in \text{Aut}(G, c)$ preserves the bijection $f$, that is, if $v \in K$ then $\phi(v) = w \iff \phi(f^{-1}(v)) = f^{-1}(w)$. Therefore, $(G, c)$ contains an $\text{PFPF}_b$ automorphism if and only if $(G', c')$ contains an $\text{PFPF}_{b'}$ automorphism.
\end{proof}

\begin{lemma} \label{lem:iso_p4_sparse}
    Let $G$ and $H$ be $P_4$-sparse graphs. The problem of determining whether $G$ is isomorphic to $H$ can be solved in polynomial time.
\end{lemma}

\begin{theorem}\label{thm:fpfaut_p4_sparse}
    Let $G$ be a $P_4$-sparse graph. The problem of determining $G$ contains a fixed point free automorphism can be solved in polynomial time.
\end{theorem}

\section{Fixed point free involutions} \label{sec:involution}

An automorphism $\phi$ of a graph is called an \emph{involution} if it has order two. The problem of finding involutions without fixed points (\pname{FPFInv}) is also \NPh~\cite{Lubiw1981}. Actually, in the graph constructed by Lubiw to show that \pname{FPFAut} is \NPh, every FPF automorphism is an FPF involution, which implies that the following problem is \NPh.

\probl{FPFInv}{A graph $G$.}{Does $G$ admit an involution without fixed points?}
In this section, we extend the results from FPF automorphisms to FPF involutions.

\subsection{NP-Hard classes and fixed edge free involutions}

\begin{theorem}\label{thm:fpfinv_main_negative_result}
    \pname{FPFInv} is \NPh\, for split, bipartite, $k$-subdivided, and $H$-free graphs, if $H$ is not an induced subgraph of $P_4$.
\end{theorem}
\begin{proof}
Let $G = (V, E)$ be a graph and let $S_k(G')$ be the $k$-subdivided graph constructed from $G$ as described in Theorem \ref{thm:subdvided}. We show that $S_k(G')$ has an FPF involution if and only if $G$ also has an FPF involution. The argument follows similarly for the other constructions.

Let $\phi_k \in \text{Aut}(S_k(G'))$, recall that $V \subseteq V(S_k(G'))$ and that $\phi_{k|V}$ induces an automorphism $\phi \in \text{Aut}(G)$, so if $\phi_k$ is an FPF involution, then $G$ contains an FPF involution. 

Now, let $\phi \in \text{Aut}(G)$ be an FPF involution, and let $\phi_k$ be constructed from $\phi$ just as described in Theorem \ref{thm:subdvided}. Since $\phi$ is isomorphic to $\phi_{k|V}$, we consider the action of $\phi_k$ at a vertex $(u, v)_{i, d}$ with $u, v \in V$, $i \in \{1, 2\}$ and $d \in [k]$. By construction, $\phi_k((u, v)_{i, d}) = (u', v')_{j, d}$ with $\phi(u) = u'$, $\phi(v) = v'$ and $j \neq i$. Since $\phi$ is an involution, $\phi_k$ is also an involution.
\end{proof}

We present an application for the constructions used in Section~\ref{sec:Negative_Results}. In the graph avoidance game, two players alternately color edges of a graph $G$ in red and blue, respectively. The player who first creates a monochromatic subgraph isomorphic to a forbidden graph $F$ loses. If $G$ admits a fixed \emph{edge} free involution $\phi$, then the second player can win by always coloring the edge $\{\phi(u), \phi(v)\}$ in blue after the first player color the edge $\{u, v\}$ in red, this is called the \emph{automorphism-based strategy}. Harary et al.~\cite{harary2002symmetric} prove that the problem of determining whether a graph admits such involution is \NPc, implying that the automorphism-based strategy for the graph avoidance game is \NPc\ to realize. The graph from Construction~\ref{con:bipartite} admits an FPF involution if and only if it has a fixed edge free involution. This provides an alternative proof of Harary's result.

\begin{corollary}[\cite{harary2002symmetric}, Theorem 5.1]\label{cor:harary}
    Deciding if a graph admits a fixed edge free involution is \NPc. 
\end{corollary}

\subsection{Algorithm for \pname{FPFInv}}

Algorithm~\ref{alg:fpf_aut} can be easily adapted to compute FPF involutions. The main change is the replacement of the \pname{PFPFAut} problem by the \pname{PFPFInv} problem.

\probl{PFPFInv}{A graph $G$, a coloring $c$, and a boolean function $b$.}{Does $(G, c)$ admit an PFPF$_b$ involution?}

\begin{algorithm2e}[ht]
      \SetAlgoLined
      \SetKwFunction{FhasFPFInv}{hasFPFInvolution}
      \SetKwProg{Fn}{}{}{}
      \Input{A graph $G$.}
      \Output{Does $G$ admit an FPF involution?}

      \Fn{\FhasFPFInv{$G$}}{
        \uIf {$|V(G)| = 1$} {
            \Return false\\
        }
        let $\mathcal{P} = \{M_1, \dots M_k\}$ be the maximal modular partition of $G$.\\
        let $(G_{/\mathcal{P}}, c)$ be a quotient colored graph associated with $\mathcal{P}$.\\
        let $b \in \{0, 1\}^k$ with $b_i = 1$ if and only if \FhasFPFInv{$G[M_i]$}=true.\\
        \uIf{($G_{/\mathcal{P}}$, $c$) has an PFPF$_b$ involution}{
            \Return true\\
        }
        \Return false\\
      }
      \caption{\FuncSty{hasFPFInvolution}}
      \label{alg:fpf_inv}
\end{algorithm2e}

We follow the same analysis described in Section~\ref{sec:efficient}.

\begin{theorem}\label{thm:fpf_inv_alg}
    Let $G$ be a graph. Algorithm~\ref{alg:fpf_inv} returns ``true'' if and only if $G$ contains a fixed point free involution.
\end{theorem}
\begin{proof}
    The proof is similar to that of Theorem~\ref{thm:fpf_aut_alg}. Let $\mathcal{P} = \{M_1, \dots M_k\}, (G_{/\mathcal{P}}, c)$ and $b$ be defined as in Algorithm~\ref{alg:fpf_inv}.
    
    Assume Algorithm~\ref{alg:fpf_inv} returns ``true''. Let $\phi_i \in \text{Aut}(G[M_i])$ and let $\psi \in \text{Aut}(G_{/\mathcal{P}}, c)$. The automorphism $\phi'$ constructed in Lemma~\ref{lem:fpf_aut_true_implies_correct} from the $\phi_i$'s and from $\psi$ is equal to the least common multiple of the order of the composed automorphisms, as their non-trivial orbits are disjoint. Since each automorphism has order at most 2, $\phi'$ is an involution.
        
    Let $\phi$ be an FPF involution of $G$. Then, the induced automorphisms $\psi \in \text{Aut}(G_{\mathcal{P}}, c)$ and $\phi_i \in \text{Aut}(G[M_i])$ must also be involutions, which implies that Algorithm~\ref{alg:fpf_inv} returns ``true''. 
\end{proof}

\begin{definition}
    Let $\mathcal{G}$ be a graph class. Let $f_{\text{pfpfi}}^{\mathcal{G}}$ denote the complexity of deciding if a graph $G \in \mathcal{G}$ admits an $\text{PFPF}_b$ involution, for an arbitrary coloring $c : V(G) \to [|V(G)|]$ and arbitrary boolean function $b : V(G) \to \{0, 1\}$.
\end{definition}

\begin{corollary} \label{cor:fpf_inv_complexity}
 Let $\mathcal{G}$ be an hereditary graph class. The complexity of deciding if an $n$-vertex graph $G \in \mathcal{G}$ contains an FPF involution is upper bounded by 
 \[n^{\mathcal{O}(1)} \cdot \left(f_{\text{iso}}^{\mathcal{G}}(n) +f_{\text{pfpfi}}^{\mathcal{G}_{/\mathcal{P}}}(\text{n})\right).\]
\end{corollary}

\subsection{Equitable partitions}

Let $G = (V, E)$ be a graph, and let $\mathcal{P} = \{V_1, V_2, \dots, V_m\}$, be a partition of $V$. We refer to the subsets $V_i$ as \emph{cells}. A partition is called \emph{equitable} (or \emph{regular}) if for all $i,j$, each vertex in $V_i$ has the same number of neighbors in $V_j$. We call a partition \emph{$k$-homogeneous} if every cell has size $k$. Equitable partitions are a versatile tool that have been used in many different fields of combinatorics, for example for deriving sharp eigenvalue bounds on the independence number~\cite{hoffman},  constructing self-orthogonal codes~\cite{codes} and clustering in various types of networks~\cite{clustering, epidemics}. The following lemma by Abiad et al.\ establishes a one-to-one correspondence between 2-homogeneous equitable partitions and FPF involutions.

\begin{lemma}[{\cite{abiad2022characterizing}, Lemma 17}]
    \label{lem:automorphismIFFpartitionRegular}
    Let~$G$ be a graph on $n$ vertices.
    Then $G$ has an automorphism being an involution without fixed points if and only if~$G$ admits an equitable partition
    with~$\frac{n}{2}$ cells each having size~2.
\end{lemma}

Lemma~\ref{lem:automorphismIFFpartitionRegular} implies that it is \NPh\ to decide whether an arbitrary graph admits a 2-homogeneous equitable partition~\cite{abiad2022characterizing}. Despite the \NP-hardness result, Abiad et. al. prove that it can be determined in polynomial time whether a cograph admits a 2-homogeneous equitable partition~\cite{abiad2022characterizing}. This was achieved using a characterization of cographs in terms of twin classes. By extending our results to FPF involutions, we achieve a generalization of Abiad's result.

\begin{theorem}\label{thm:main_corollary_from_Aida}
The problem of determining whether a graph admits an equitable partition with $\frac{n}{2}$ parts of size $2$ can be solved in polynomial time when restricted to the class of bounded modular-width graphs, tree-cographs and $P_4$-sparse graphs.
\end{theorem}
\begin{proof}
    It suffices to check that complete graphs, trees and cyclops spiders admit an efficient algorithm for the \pname{PFPFInv} problem. For complete graphs and cyclops spiders it follows directly since the compositions of automorphisms described in Lemmas~\ref{lem:pfpf_complete} and~\ref{lem:pfpf_spider} creates an involution if the initial automorphisms are involutions. 
    
    Let $T_r$ be a rooted tree, and let $C_v$ be the set of children of a vertex $v$.  Let $\phi \in \text{Aut}(T_r)$, note that either $\phi(C_v) = C_v$ or $\phi(v) \neq v$. In Lemma \ref{lem:pfpf_trees} we solve \pname{PFPFAut} in rooted trees by constructing a complete graph from each set $C_v$ and solving \pname{PFPFAut} in this graph class. Since $\phi$ is an involution if and only if $\phi_{|C_v}$ is an involution every time $C_v$ is fixed, we conclude that the same reduction works to solve \pname{PFPFInv} for rooted trees. For general trees, notice that if a tree has two centers and there is an FPF automorphism that swaps the centers, then there exists an FPF involution that does the same. 
\end{proof}

\subsection{Compact Graphs} \label{subsec:compact}

We finish this section by observing a difference between FPF automorphisms and FPF involutions.
A graph is called compact if all extreme points of the convex hull of the double stochastic matrices that commute with its adjacency matrix are 01-matrices. In 1991, Tinhofer~\cite{tinhofer1991note} showed that the isomorphism problem can be solved in polynomial time for compact graphs by finding basic solutions to a linear program. The linear programming approach can be easily modified to determine the existence of automorphisms without fixed points, leading to the following result.

\begin{theorem} \label{thm:compact}
    Let $G$ be a compact graph. The problem of determining $G$ contains a fixed point free automorphism can be solved in polynomial time.
\end{theorem}
\begin{proof}
Let $A := A(G)$ be the adjacency matrix of $G$ and consider the following linear program.

\begin{equation} 
\boxed{
\begin{array}{rll}
 {\tt maximize} &  \sum_{i,j} x_{i,j}\\
 {\tt subject\ to}  
    & AX = XA\\
    & \sum_{j} x_{i,j} = 1, \quad & \forall i \\
    & \sum_{i} x_{i,j} = 1, \quad & \forall j \\
    & \sum_{i} x_{i,i} = 0, \\
    & 0 \leq x_{i,j} \leq 1, \quad & \forall i, j \\
\end{array}
}
\label{LP}
\end{equation}

The constraints ensure that any feasible solution is a doubly stochastic matrix that commutes with the adjacency matrix and has $0$ diagonal. Let $X$ be any such matrix, since $G$ is compact, $X$ can be written as $X = \sum_{i} a_i P_i$ with $a_i > 0$, $\sum a_i = 1$ and $P_i$ corresponding to some automorphism of $G$. Since $X$ has $0$ diagonal, each $P_i$ must also have $0$ diagonal, so the extreme points of the polytope described by the linear program~\eqref{LP} correspond to FPF automorphisms. 

Since the Simplex method for LPs over 0/1 polytopes finds extreme points in polynomial time~\cite{black2024simplex}, we have a polynomial-time algorithm for \pname{FPFAut} on compact graphs.
\end{proof}

For FPF involutions, the situation is different and so far, no polynomial-time algorithm is known for this problem. We observe that the restriction of forcing the matrix $X$ to be symmetric in formulation \eqref{LP} does not imply that the extreme points of the polytope are symmetric permutations (consider, for example, an odd cycle). The class of strong-tree-cographs (a subclass of tree-cographs) is compact~\cite{tinhofer1988strong}, so Theorem~\ref{thm:main_corollary_from_Aida} suggests a positive answer to the following question.
 
\begin{question}
    Can \pname{FPFInv} be solved in polynomial time for compact graphs?
\end{question}

\section{Conclusion} \label{sec:conclusion}
Lubiw~\cite{Lubiw1981} introduced the \pname{FPFAut} problem, which asks whether a given graph admits a fixed point free automorphism, and proved that the problem is \NPc.
In this paper, we investigate the \pname{FPFAut} problem when restricted to graph classes, showing that it remains \NPc\ for many of them.
We also present an algorithm based on modular decomposition and demonstrate its efficiency on three classes.
Finally, we extend our results to the \pname{FPFInv} problem.

A central open problem in parameterized complexity is determining the complexity of \pname{GraphISO} parameterized by rank-width~\cite{NEUEN2026100918}. Currently, the problem is known to be in XP with respect to this parameter~\cite{grohe2015isomorphism}. In Section~\ref{subsec:modular_width}, we use this result to show that \pname{FPFAut} can be solved in polynomial time for graphs of bounded modular-width. It would be interesting to extended this approach to graphs of bounded rank-width.

\begin{question}
    Can \pname{FPFAut} be solved in polynomial time for graphs of bounded rank-width?
\end{question}

To further motivate the study of fixed-point-free automorphisms, we recall a problem introduced by Nagoya~\cite{nagoya2009computing}. In 2009, Nagoya proved that given a graph with a non-trivial automorphism, finding a pair of vertices $(u, v)$ where $u$ is mapped to $v$ by a non-trivial automorphism is computationally equivalent to computing a non-trivial automorphism of the graph. Nagoya's proof relies on the construction of a new graph $G'$ from the original graph $G$ and the pair of vertices $(u, v)$, such that $G'$ admits a non-trivial automorphism if and only if $G$ admits one, with the condition $G' < G$ in some partial order. The author asks whether the same result holds for FPF automorphisms. As far as we know, this problem remains open and the main difficulty is that Nagoya's construction creates FPF automorphisms that do not exist in the original graph. It would be interesting to find a construction that preserves the one-to-one correspondence between the FPF automorphisms.

\subsection*{Acknowledgements}

Aida Abiad is partially supported by NWO (Dutch Research Council) through the grant VI.Vidi.213.085. Gabriel Coutinho, Emanuel Juliano and Vinicius Santos are supported by FAPEMIG and by CNPq.

 
\DeclareRobustCommand{\VAN}[3]{#3}

\bibliographystyle{abbrv}

\IfFileExists{references.bib}
{\bibliography{references.bib}}
{\bibliography{../references}}

\end{document}